\documentclass[12pt,journal,final,onecolumn]{IEEEtran}

\usepackage{graphicx}
\usepackage[cmex10]{amsmath}
\usepackage{amssymb}
\usepackage{amsthm}
\usepackage{amsfonts}
\usepackage{bm}
\usepackage{cite}
\usepackage{url}
\usepackage[svgnames]{xcolor} 
\usepackage{pstricks,pst-node,pst-plot,pstricks-add}
\usepackage[tight,footnotesize]{subfigure}
\usepackage{algorithm}
\usepackage{algpseudocode}

\graphicspath{{figs/}}

\input{niesen.def}

\begin{document}

\title{Coded Caching with Nonuniform Demands}

\author{Urs Niesen and Mohammad Ali Maddah-Ali%
\thanks{Urs Niesen was with Bell Labs, he is now with the Qualcomm NJ
Research Center. Mohammad Ali Maddah-Ali is with Bell Labs, Nokia.
Emails: urs.niesen@ieee.org, mohammadali.maddah-ali@nokia.com}%
\thanks{The first version of this paper was posted on arXiv in Aug. 2013.}%
\thanks{The work of Urs Niesen was supported in part by AFOSR under grant FA9550-09-1-0317.}%
}

\maketitle

\begin{abstract} 
    We consider a network consisting of a file server connected through
    a shared link to a number of users, each equipped with a cache.
    Knowing the popularity distribution of the files, the goal is to
    optimally populate the caches such as to minimize the expected load
    of the shared link. For a single cache, it is well known that
    storing the most popular files is optimal in this setting. However,
    we show here that this is no longer the case for multiple caches.
    Indeed, caching only the most popular files can be highly
    suboptimal. Instead, a fundamentally different approach is needed,
    in which the cache contents are used as side information for coded
    communication over the shared link. We propose such a coded caching
    scheme and prove that it is close to optimal.
\end{abstract}

\section{Introduction}
\label{sec:intro}

Caching or prefetching is a technique to reduce network load by storing
part of the content to be distributed at or near end users. In this
paper, we design near-optimal caching strategies for a basic network
scenario, introduced in~\cite{maddah-ali12a}, consisting of one server
connected through a shared, error-free link to $K$ users as illustrated
in Fig.~\ref{fig:setting}. The server has access to $N$ files each of
size $F$ bits. Each user has an isolated cache memory of size $MF$ bits.  

\begin{figure}[htbp]
    \centering%
    \includegraphics{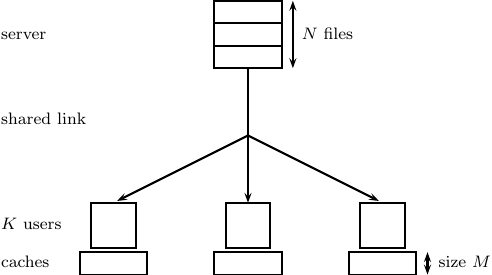}
    \caption{The caching problem with $N=3$ files and $K=3$ users for
    normalized cache size $M=1$.} 
    \label{fig:setting}
\end{figure}

During times of low traffic demand (for example in the early morning) or
when connected to a network with large available bandwidth (for example
a mobile handset connected to WiFi), users can save some part of the
files to their local caches. During a later time, when a user requests
one of the files, the local cache can be used to reduce network load.
More formally, the system operates in two distinct phases: a
\emph{placement phase} and a \emph{delivery phase}. In the placement
phase, each user can save part of the $N$ files in its cache memory.  In
the delivery phase, each user randomly requests one of the files in the
library independently of the other users and with identical
distribution. The server is informed of these requests and proceeds by
transmitting a message over the shared link.  Each user then aims to
reconstruct its requested file from the content of its cache and the
message received over the shared link in the delivery phase. The
placement and delivery phases of the system should be designed to
minimize the load of the shared link subject to the memory constraint in
the placement phase and the reconstruction constraint in the delivery
phase.

Designing and analyzing caching systems for such (and more complicated)
networks has a long history, see for
example~\cite{leff93,korupolu99,baev01,meyerson01,baev08,borst10}.  The
impact of specific file popularities (such as Zipf or other heavy-tail
distributions) on the performance of caching has been analyzed
in~\cite{breslau99, wolman99, hefeeda08}, among others. These papers
consider \emph{uncoded} caching. For the basic network scenario
considered here with only a \emph{single} cache ($K=1$), it turns out
that such uncoded caching strategies are optimal.  Indeed, the optimal
strategy in this case is the highest-popularity first (HPF) caching
scheme\footnote{The HPF strategy is the offline equivalent of the
well-known and commonly-used least-frequently used (LFU) online cache
eviction policy, in which, when a new file is requested, the file that
is least-frequently used is evicted from the cache. If the actual file
popularities are known, as is assumed here, then this reduces to caching
the $M$ most popular files, i.e., the HPF caching scheme.}, in which
each user caches the $M$ most popular files in its cache (see
Appendix~\ref{sec:HPF} for a proof).

In this paper, we show that this intuition for a \emph{single} cache
does not carry over to \emph{multiple} caches ($K > 1$). In fact, HPF
can be arbitrarily suboptimal in the multi-cache setting.  Instead, a
fundamentally different approach to the caching problem is required. We
propose the use of a \emph{coded} caching scheme, recently introduced by
the present authors in~\cite{maddah-ali12a, maddah-ali13}.  These coded
caching schemes work by carefully designing the placement phase so as to
enable a simultaneous coded multicasting gain among the users, even if
they request different files, resulting in significantly better
performance compared to uncoded schemes. The following toy example
illustrates this approach.

\begin{example}
    \label{eg:toy}
    Consider a scenario as depicted in Fig.~\ref{fig:setting} with $N=2$
    files, say $A$ and $B$, with popularities $p_A=2/3$ and $p_B=1/3$.
    Assume we have $K=2$ users and a normalized memory size $M=1$. 

    In this setting, HPF uses the entire memory to cache the more
    popular file $A$. The classic analysis of this system assumes
    unicast delivery, in which case the link load of HPF is given by $K$
    times the cache miss rate times the file size $F$. Here the miss rate
    is $p_B$, and the resulting expected link load of HPF under unicast
    delivery is $2/3F$ bits. If we allow broadcast delivery, then the
    expected link load of HPF is $(1-p_A^2)F = 5/9F$.

    We next describe the coded caching scheme from~\cite{maddah-ali12a}.
    We split each file into two parts of equal size so that $A = (A_1,
    A_2)$ and $B = (B_1, B_2)$. Store $A_1, B_1$ at the first user and
    $A_2, B_2$ at the second user. Assume user one requests file $A$ and
    user two requests file $B$. These requests can then be satisfied
    with a single coded transmission $A_2\oplus B_1$ of size $1/2F$ bits
    from the server, where $\oplus$ denotes bit-wise XOR. Using the
    information stored in their caches, each user can recover its
    requested file. The other requests can be satisfied in a similar
    manner, leading to an expected link load of $1/2F$ bits for the
    coded caching scheme.

    This illustrates that, while HPF minimizes the link load for a
    single cache, this is no longer the case for multiple caches and
    that coding is required in these situations. Furthermore, it shows
    that the miss rate, which is minimized by HPF, is no longer the
    appropriate metric for scenarios with multiple caches.
\end{example}

The coded caching schemes from~\cite{maddah-ali12a, maddah-ali13} were
shown there to approximately minimize the \emph{peak} load (i.e., the
load for the worst-case user requests) of the shared link. However, in
many situations the file popularities differ over many orders of
magnitudes (see Fig.~\ref{fig:netflix_p} in Section~\ref{sec:empirical}
for a real-life example), and hence the \emph{expected} load is more
relevant than the peak load.  This is especially the case for systems
having many users or running over long time periods, for which the law
of large numbers takes effect. 

In this paper, we investigate caching systems under this expected load
performance metric. In particular, we propose a coded caching scheme
that is able to handle widely differing file popularities. A key
ingredient in this scheme is the idea of grouping files with similar
popularities together.  We provide a bicriteria (with respect to cache
memory and link load) approximation guarantee for this grouped coded
caching scheme. The proof of approximate optimality of the proposed
scheme links the optimal expected load to the optimal peak load and is
quite intricate. We apply the proposed scheme to the file popularities
of the Netflix video catalog and show that it can significantly
outperform the baseline uncoded HPF scheme.

The remainder of this paper is organized as follows.
Section~\ref{sec:problem} formally defines the problem setting.
Section~\ref{sec:background} provides background information on coded
caching. Sections~\ref{sec:theoretical} and~\ref{sec:empirical} present
the main results of this paper. Section~\ref{sec:conclusion} contains
concluding remarks. All proofs are delegated to the appendix.

\section{Problem Setting}
\label{sec:problem}

Consider again the network depicted in Fig.~\ref{fig:setting} in
Section~\ref{sec:intro} with one server connected through a noiseless
broadcast channel to $K$ users. The server has a library of $N$ files 
with indices
\begin{equation*}
    \mc{N} \defeq \{1, 2, \dots, N\}
\end{equation*}
each of size $F$ bits. Each user is equipped with a cache memory of size
$MF$ bits. We point out that, while the cache size $MF$ is always an
integer number of bits, the normalized cache size $M$ is in general a
nonnegative real number. 

As mentioned in the introduction, the system operates in two phases, a
placement and a delivery phase. The operation of the placement phase is
specified by a caching function, which describes the contents cached at
each user subject to the memory constraint of $MF$ bits. 

During the subsequent delivery phase, each user requests one of the $N$
files. Denote by $d_k\in\mc{N}$ the request of user $k$, and denote by
$\underline{d}$ the corresponding vector $(d_k)_{k=1}^K$ of all user requests. For
each such user request vector, the delivery phase is specified by an
encoding function at the transmitter and $K$ decoding functions, one for
each receiver. The encoding function maps the file library to a channel
input of size $R_{\underline{d}}F$ bits, for some real number $R_{\underline{d}}$. The decoding
function at user $k$ maps the channel output (which, due to the
noiseless broadcast, is identical to the channel input) and that user's
cache content to an estimate of its requested file $d_k$. The encoding
and decoding functions need to guarantee that each user is able to
correctly recover its requested file. 

The quantity $R_{\underline{d}}F$ is the load for user requests
$\underline{d}$ and $R_{\underline{d}}$ is the corresponding rate.
Observe that we may have different encoding and decoding functions for
different user request vectors and consequently different rates
$R_{\underline{d}}$. The collection of one caching function,
$N^K$ encoding functions (one for each of the $N^K$ possible request
vectors $\underline{d}$), and $KN^K$ decoding functions (one for each of
the $K$ users and for each of the $N^K$ possible request vectors
$\underline{d}$) form a caching scheme. 

We next specify a probabilistic model for the user requests. Denote
by $p_n$ the popularity of file $n$, so that the $p_n$ form a distribution over
$\mc{N}$. Assume that each user independently requests file $n\in\mc{N}$
with probability $p_n$. This request model is called the independent
reference model in the operating systems caching literature
\cite[Chapter~6.6]{coffman73}. The \emph{expected rate} of a caching
scheme under file popularity $\{p_n\}_{n \in \mc{N}}$ is then
\begin{equation*}
    \sum_{ \underline{d} \in\mc{N}^K} \biggl( \prod_{k=1}^K p_{d_k} \biggr) R_{\underline{d}}.
\end{equation*}

In the following, we will connect the expected rate of a caching scheme
to its \emph{peak rate}, which is defined as
\begin{equation*}
    \max_{\underline{d}\in\mc{N}^K} R_{\underline{d}}. 
\end{equation*}
In words, the peak rate is the rate at which every possible user request
can be satisfied. 

We say that a expected (peak) rate is achievable if for all large enough
file sizes $F$ there exists a caching scheme of that expected (peak)
rate. The optimal expected (peak) rate is the infimum of all achievable
expected (peak) rates. Note that the optimal rate is an asymptotic
quantity, describing the optimal system performance in the large-file
limit. In the remainder of the paper, we are interested in
characterizing the optimal expected rate. For the $K$-user cache network
with $N$ files in $\mc{N}$ with popularities $\{p_n\}_{n \in \mc{N}}$,
and normalized cache size $M$, we denote this optimal expected rate by
$R^\star(M,\mc{N},K, \{p_n\})$.

\section{Background on Coded Caching}
\label{sec:background}

For future reference, we define the function $r(M, N, K)$ as 
\begin{equation}
    \label{eq:PeakRate}
    r(M, N, K) 
    \defeq 
    \begin{cases}
        K\cdot(1-M/N) \cdot\min\biggl\{\frac{N}{KM}\Bigl(1-(1-M/N)^K\Bigr),\; \frac{N}{K} \biggr\} 
        & \text{for $0 < M \leq N$}, \\
        \min\{N,K\} & \text{for $M = 0$}, \\
        0 & \text{for $M > N$}.
    \end{cases}
\end{equation}
In~\cite{maddah-ali12a, maddah-ali13}, we have shown that in a system
with $N$ files, $K$ users, and cache memory of normalized size $M$, a
peak rate of   $r(M, N, K)$  is achievable with high probability\footnote{The statement is
probabilistic due to the randomness of the placement phase (see
Algorithm~\ref{alg:1}).} for large enough file size $F$. We now
briefly describe how the peak rate~\eqref{eq:PeakRate} can be achieved;
the discussion here follows~\cite{maddah-ali13}.

In the placement phase, each user saves a random subset of $MF/N$ bits
of each file into its cache memory. These random subsets are chosen
uniformly and independently for each user and file.  Since there are a
total of $N$ files, this satisfies the memory constraint of $MF$ bits.
In the delivery phase, after the users' requests are revealed, the
server delivers the requested files while maximally exploiting the side
information available in each user's cache. This is done by coding
several requested files together.

\begin{algorithm}[h!]
    \caption{Coded caching scheme from~\cite{maddah-ali13} achieving
    peak rate~\eqref{eq:PeakRate}} 
    \label{alg:1}
    \begin{algorithmic}[0]
        \Procedure{Placement}{}
        \For{$k\in\mc{K}, n\in\mc{N}$}
        \State User $k$ caches a random $\tfrac{MF}{N}$-bit subset of file $n$ \label{alg:1_cache}
        \EndFor
        \EndProcedure
        \Statex
        \Procedure{Delivery}{$d_1,\dots,d_K$}
        \For{$s=K, K-1, \ldots, 1$} \label{alg:1_sloop}
        \For{$\mc{S}\subset\mc{K}: \card{\mc{S}}=s$} \label{alg:1_Sloop}
        \State Server sends \(\oplus_{k\in\mc{S}} V_{k,\mc{S}\setminus\{k\}}\) \label{alg:1_send}
        \EndFor 
        \EndFor 
        \EndProcedure
        \Statex
        \Procedure{Delivery'}{$d_1,\dots,d_K$}
        \For{$n\in\mc{N}$}
        \State Server sends enough random linear combinations of bits in file $n$ for all users requesting it to decode \label{alg:1_send2}
        \EndFor
        \EndProcedure
    \end{algorithmic}
\end{algorithm}

The placement and delivery procedures are formally stated in
Algorithm~\ref{alg:1}. In the description of Algorithm~\ref{alg:1},
$\mc{K}$ and $\mc{N}$ denote the sets $\{1,2,\dots,K\}$ and $\{1,2,\dots,N\}$,
respectively. Furthermore, for a subset $S\subset\mc{K}$ of users,
$V_{k,\mc{S}}$ denotes the vector of file bits that are requested by
user $k$ and that are available exclusively in the cache of every user
in $\mc{S}$ and missing in the cache of every user outside $\mc{S}$.
The symbol $\oplus$ denotes again bit-wise XOR, where the vectors
$V_{k,\mc{S}}$ belonging to the same $\oplus$ are assumed to be zero
padded to common length.  Algorithm~\ref{alg:1} contains two possible
delivery procedures; the server chooses whichever one results in smaller
rate.

The following example from~\cite{maddah-ali13} illustrates the
algorithm. 
\begin{example}[\emph{Illustration of Algorithm~\ref{alg:1}}]
    \label{eg:illustration}
    We consider the caching problem with $N=2$ files $A$, $B$, with
    $K=2$ users, and with normalized memory size $M\in(0,2]$. In the
    placement phase of Algorithm~\ref{alg:1}, each user caches a random
    subset of $MF/2$ bits of each file.
    
    Focusing on file $A$, we see that the placement procedure implicitly
    partitions this file into four subfiles
    \begin{equation*}
        A=(A_\emptyset, A_{1},  A_{2}, A_{1,2}),
    \end{equation*}
    where for each subset $\mc{S}\subset\mc{K}$, $A_{\mc{S}}$ denotes the
    bits of file $A$ that are stored exclusively in the cache memories
    of users in $\mc{S}$.\footnote{To simplify the
    exposition, we slightly abuse notation by writing $A_{1}, A_{1,2},
    \dots$ for $A_{\{1\}}, A_{\{1,2\}}, \dots$} For example, $A_{1,2}$ are the
    bits of file $A$ stored in the cache of users one and two, $A_{2}$
    are the bits of $A$ stored exclusively in the cache of user two, and
    $A_\emptyset$ are the bits of file $A$ not stored in the cache of
    either user. For large enough file size $F$, the law of large
    numbers guarantees that
    \begin{equation*}
        \card{A_{\mc{S}}}/F
        \approx (M/2)^{\card{\mc{S}}}(1-M/2)^{2-\card{\mc{S}}}
    \end{equation*}
    with high probability and similarly for file $B$. 

    Consider next the delivery procedure. It can be verified that in this
    setting the first delivery procedure is better and will hence be
    used by the server. Assume users one and two request files $A$ and
    $B$, respectively. In this case, $V_{1,\{2\}}=A_{2}$, $V_{2,\{1\}}=B_{1}$,
    $V_{1,\emptyset}=A_\emptyset$, and $V_{2,\emptyset}=B_\emptyset$.
    Hence, the server sends $A_{2} \oplus B_{1}$, $A_\emptyset$, and
    $B_\emptyset$ over the shared link.

    The file parts $A_\emptyset$ and $B_\emptyset$ are not cached at any
    of the users, and hence they obviously have to be sent from the
    server for successful recovery of the requested files.  The more
    interesting transmission is $A_{2} \oplus B_{1}$.  Observe that user
    one has $B_{1}$ stored in its cache memory.  Hence, user one can
    solve for the desired file part $A_{2}$ from the received message
    $A_{2} \oplus B_{1}$. Similarly, user two can solve for desired file
    part $B_{1}$ using $A_{2}$ stored in its cache memory.  In other
    words, the transmission $A_{2} \oplus B_{1}$ is simultaneously
    useful for both users. Thus, even though the two users request
    different files, the server can successfully multicast useful
    information to both of them. The rate of the messages sent by the
    server is
    \begin{equation*}
        (M/2)(1-M/2)+2(1-M/2)^2
        = 2\cdot (1-M/2)\cdot \frac{1}{M}\bigl( 1-(1-M/2)^2 \bigr)
        = r(M,2,2).
    \end{equation*}

    While the analysis here was for file requests $(A,B)$, the same
    arguments hold for all other possible file requests $(B,A)$,
    $(A,A)$, and $(B,B)$ as well. In each case, the side information in
    the caches is used to create coded multicasting opportunities for
    users with (possibly) different demands. In other words, the content
    placement is performed such as to enable coded multicasting
    opportunities \emph{simultaneously} for all possible demands. The
    rate obtained above holds therefore for every possible user demands,
    i.e., it is an achievable peak rate for the caching problem.
\end{example}

The problem setting and Algorithm~\ref{alg:1} make
several idealized assumptions such as equal file size, synchronous user
requests, large file size, and so on, in order to simplify the
theoretical analysis. These assumptions can be relaxed as is discussed
in~\cite{maddah-ali13}. A working video-streaming prototype using the
coded caching approach and dealing with all these issues is presented
in~\cite{niesen14}.

We also point out that the caching problem is related to
the index coding problem~\cite{birk06, bar-yossef11} and the network
coding problem~\cite{ahlswede00}. The connection and differences between
these problems is described in detail in \cite{maddah-ali12a}.

\section{Theoretical Results}
\label{sec:theoretical}

In~\cite{maddah-ali13}, we prove that the rate $r(M, N, K)$ defined
in~\eqref{eq:PeakRate} of the caching scheme reviewed in
Section~\ref{sec:background} is within a constant factor of the optimal
\emph{peak} rate. In this paper, we are instead interested in the
\emph{expected} rate. As a corollary to the results presented later in
this section, we show that the same rate $r(M, N, K)$ is also within a
constant factor of the optimal \emph{expected} rate for uniform file
popularities, i.e., $p_1= p_2= \dots= p_N$.  The important question is
how to achieve good performance in the more realistic case of files
having popularities varying over several orders of magnitude.

Before explaining the proposed algorithm for such cases, we first
highlight two important features of Algorithm~\ref{alg:1}. The first
feature is that the delivery algorithm exploits coded multicasting
opportunities among every subset of users.  The second feature is the
symmetry in the placement phase. It is this symmetry that permits to
easily identify and quantify the coded multicasting opportunities.
These two features together are crucial for the approximate optimality
of Algorithm~\ref{alg:1} for uniform file popularities.

Consider now some of the options for nonuniform file popularities. One
option is to apply the placement procedure of Algorithm~\ref{alg:1} to
all $N$ files. The advantage of this scheme is that the symmetry of the
content placement is preserved, and therefore, in the delivery phase, we
can again identify and quantify the coded multicasting opportunities
among all $N$ files. The disadvantage is that the difference in the
popularities among the $N$ cached files is ignored. Since these files
can have widely differing popularities, it is wasteful to dedicate the
same fraction of memory to each of them. As a result, this approach
may not perform well.

Another option is to dedicate a different amount of memory to each file
in the placement phase. For example, the amount of allocated memory could
be proportional to the popularity of a file. While this option takes the
different file popularities into account, it breaks the symmetry of the
content placement.  As a result, the delivery phase becomes difficult to
analyze.

Here we propose the idea of \emph{grouping} as an alternative solution
that has the advantages of both of theses approaches and can be proved
to be approximately optimal. In the proposed scheme, we partition the
files into groups with approximately uniform popularities (see
Fig.~\ref{fig:grouping} for a representative example).
In the placement phase, we ignore differing popularities of files within
the same group and allocate to each of these files the same amount of cache
memory.  However, files in different groups may have different memory
allocations. In the delivery phase, the demands within each group are
delivered using Algorithm~\ref{alg:1}, while ignoring coding
opportunities across different file groups. Note that, since the
symmetry within each group has been preserved, the delivery phase is
analytically tractable. Moreover, since different groups have different
memory allocations, we can use more memory for files with higher
popularity. The size of the groups and the amount of memory allocated to
each group can be optimized to minimize the expected rate over the
shared link.

\begin{figure}[htbp]
    \centering%
    \includegraphics{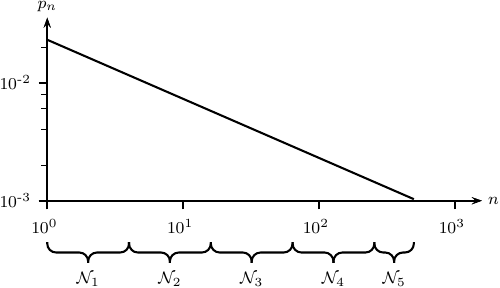}
    \caption{Sample file popularities $p_n$ and 
        file grouping. The figure depicts a Zipf file
        popularity with exponent $-1/2$ over $N=500$ files (see
        Example~\ref{eg:zipf} below for a formal definition of Zipf
        popularity). The files
        are partitioned into five groups $\mc{N}_1, \mc{N}_2, \dots,
        \mc{N}_5$ such that the files within the same group have
        approximately uniform popularity. Each file in the same group is
        allocated the same amount of cache memory, but files in different 
        groups can have different memory allocations.} 
    \label{fig:grouping}
\end{figure}

We now describe the proposed scheme in detail. We partition the $N$
files $\mc{N}$ into $L$ groups $\mc{N}_1, \mc{N}_2, \dots, \mc{N}_L$.
Denote by $N_\ell$ the size of group $\mc{N}_\ell$ so that
\begin{equation*}
    \sum_{\ell=1}^L N_\ell = N.
\end{equation*}

For the placement phase, we allocate a fraction of the memory to each of
the groups $\mc{N}_\ell$ of files. Denote by $M_\ell F$ the number of
bits allocated to cache files in $\mc{N}_\ell$. $M_\ell$ must be chosen
such that the total memory constraint is satisfied, i.e., 
\begin{equation*}
    \sum_{\ell=1}^L M_\ell F = MF.
\end{equation*}

Once the memory allocation is done, we proceed with the actual placement
phase: Each user randomly selects $M_\ell F /N_\ell $ bits from each
file in group $\mc{N}_\ell$ and stores them in its cache memory. With
this, the total number of bits cached at each user from each file group
$\mc{N}_\ell$ is $N_\ell\cdot M_\ell F/N_\ell = M_\ell F$ as required.

In the delivery phase, each user randomly requests a file independently
and identically distributed according to $\{p_n\}_{n\in\mc{N}}$. Denote by
$\mc{K}_\ell$ the set of those users that request a file in the group
$\mc{N}_\ell$ of files. Note that $\mc{K}_1, \mc{K}_2, \dots, \mc{K}_L$
partitions the total of $K$ users into $L$ groups. Denote by
$\msf{K}_\ell$ the cardinality of user group $\mc{K}_\ell$. Since the
groups $\mc{K}_\ell$ depend on the random choice of the user requests,
the cardinalities $\msf{K}_\ell$ are random variables (indicated here
and throughout by the use of sans-serif font). The distribution of
$\msf{K}_\ell$ is Binomial with parameters $N$ and
$\sum_{n\in\mc{N}_\ell}p_n$. The server uses the same
delivery procedure as in Algorithm~\ref{alg:1} $L$ times, once for each
group $\mc{K}_\ell$ of users. 

The next theorem analyzes the expected rate of the proposed grouped coded caching
scheme just described for large file size $F$. This yields an upper
bound on the optimal expected rate $R^\star(M,\mc{N},K)$ for the caching
problem with arbitrary popularity distribution $\{p_n\}_{n\in\mc{N}}$. 
\begin{theorem}
    \label{thm:upper}
    Consider the caching problem with $N$ files $\mc{N}$ with arbitrary
    popularities $\{p_n\}_{n\in\mc{N}}$ and $K$ users
    each with normalized cache size $M$. Let $\mc{N}_1, \mc{N}_2, \dots,
    \mc{N}_L$ be an arbitrary partition of $\mc{N}$. Then
    \begin{align*}
        R^\star(M,\mc{N},K,  \{p_n\}) 
        & \leq \min_{\{M_\ell\}: \sum_\ell M_\ell = M}
        \sum_{\ell=1}^L \E\bigl( r(M_\ell, N_\ell, \msf{K}_\ell) \bigr) \\
        & \leq \sum_{\ell=1}^L \E\bigl( r(M/L, N_\ell, \msf{K}_\ell) \bigr),
    \end{align*}
    where the function $r(M, N, K)$ is defined in~\eqref{eq:PeakRate},
    and where the expectations are with respect to
    $\{\msf{K}_\ell\}_{\ell=1}^L$.
\end{theorem}
    
The first inequality in Theorem~\ref{thm:upper} upper bounds the optimal
expected rate $R^\star(M,\mc{N},K)$ by the rate of the proposed grouped
coded caching scheme.  Each term in the sum corresponds to the rate of
serving the users in one of the subgroups $\mc{K}_\ell$, and the sum
rate is minimized over the choice of memory allocation $M_\ell$.  The
second inequality follows from the simple memory allocation $M_\ell =
M/L$ for all $\ell$.  We point out that, even if each group is allocated
the same amount of memory $M/L$, the memory allocated to an individual
file in group $\mc{N}_\ell$ is $M/(N_\ell L)$, which varies as a
function of $\ell$ (see also Example~\ref{eg:zipf} below).

As mentioned before, the choice of file groups $\mc{N}_\ell$ can be
optimized to minimize the expected rate. Here we introduce a potentially
suboptimal choice of grouping that has however the advantage of
admitting a formal approximate optimality guarantee. By relabeling the
files, we can assume without loss of generality that $p_1 \geq p_2 \geq
\dots \geq p_N > 0$.  Let $\mc{N}_1$ be the files $\{1, 2, \dots, N_1\}$
with $N_1$ such that $p_{N_1} \geq p_1/2$ and $p_{N_1+1} < p_1/2$. Thus,
$\mc{N}_1$ are the most popular files and all files in this group have
popularity differing by at most a factor two. Similarly, define
$\mc{N}_2$ as the group of next most popular files, and so on.  In
general, for any two files $n, n'$ in the same group $\mc{N}_\ell$ the
file popularities $p_{n}$ and $p_{n'}$ differ by at most a factor two.
In other words, let $n$ be the smallest number in $\mc{N}_\ell$. Then,
\begin{align*}
    p_n \geq p_{n+N_\ell-1} & \geq p_n/2 \\
    \shortintertext{and} \\
    p_{n+N_\ell} & < p_n/2.
\end{align*}
We say that the files $\mc{N}$ are \emph{partitioned to within
popularity factor two} into $\mc{N}_1, \mc{N}_2, \dots, \mc{N}_L$. An
example of this choice of grouping is depicted in
Fig.~\ref{fig:grouping}.

This particular choice of grouping has two important features. First,
files within the same group have popularity differing by at most a
factor two. This limits the loss due to allocation of the same amount of
memory to each file within the same group. Second, since the popularity
of files in $\mc{N}_\ell$ decreases exponentially in $\ell$, the total
number $L$ of groups is small (see also the discussion in
Example~\ref{eg:zipf} below). In fact, $L = \ceil{\log p_1/p_N}$. This limits
the loss due to ignoring coding opportunities across different groups.

The next theorem establishes a lower bound on the optimal expected rate
$R^\star(M,\mc{N},K)$ for $\mc{N}$ partitioned in this manner.
\begin{theorem}
    \label{thm:lower}
    Consider the caching problem with $N$ files $\mc{N}$ with arbitrary
    popularities $\{p_n\}_{n\in \mc{N}}$ and $K$ users
    each with normalized cache size $M$. Let $\mc{N}_1, \mc{N}_2, \dots,
    \mc{N}_L$ be a partition of $\mc{N}$ to within popularity
    factor two. Then
    \begin{equation*}
        R^\star(M,\mc{N},K, \{p_n\})
        \geq  \frac{1}{c L} \sum_{\ell=1}^L \E\bigl( r(M, N_\ell, \msf{K}_\ell) \bigr),
    \end{equation*}
    where $c$ is a positive constant independent of the problem parameters,
    where the function $r(M, N, K)$ is defined in~\eqref{eq:PeakRate},
    and where the expectation is with respect to
    $\{\msf{K}_\ell\}_{\ell=1}^L$.
\end{theorem}

The value of the constant $c$ in Theorem~\ref{thm:lower} is quite large
and could be reduced by a more careful and involved analysis. For
example, it could be optimized by choosing a factor different from two
for the file grouping.

As can be seen from Theorem~\ref{thm:lower}, the specific grouping to
within popularity factor two introduced earlier is used here to develop
a lower bound on the optimal expected rate $R^\star(M,\mc{N},K,
\{p_n\})$.  We emphasize that the optimal scheme achieving
$R^\star(M,\mc{N},K, \{p_n\})$ is not restricted and, in particular, may
not use file grouping.

Recall from Theorem~\ref{thm:upper} that the expected rate of the
proposed coded caching scheme with equal memory allocation $M_\ell =
M/L$ for all $\ell$ is 
\begin{equation*}
    \sum_{\ell=1}^L \E\bigl( r(M/L, N_\ell, \msf{K}_\ell) \bigr).
\end{equation*}
Theorems~\ref{thm:upper} and~\ref{thm:lower} therefore provide a
bicriteria approximation guarantee for the performance of the proposed
grouped coded caching scheme as follows. 

Fix the parameters $K$, $\mc{N}$,  and $\{p_n\}_{n\in \mc{N}}$, and consider the set
$\mc{A}\subset\R^2$ of all achievable expected-rate--memory pairs
$(R,M)$ for the caching problem. Note that 
\begin{equation*}
    R^\star(M,\mc{N},K, \{p_n\}) = \inf \{R: (R,M)\in\mc{A}\}.
\end{equation*}
Theorems~\ref{thm:upper} and~\ref{thm:lower} then show that if a
Pareto-optimal scheme achieves a point $(R,M)$ on the boundary of the
set $\mc{A}$ of achievable expected-rate--memory pairs, then the grouped
coded caching scheme proposed here achieves at least the
expected-rate--memory pair $(cLR, LM)$. We will call this a
$(cL,L)$-bicriteria approximation guarantee in the following.
Bicriteria approximations of this type are quite common in the caching
literature.\footnote{For example, the celebrated competitive-optimality
result of the least-recently used (LRU) caching policy by Sleator and
Tarjan \cite[Theorem~6]{sleator85} proves that LRU has miss rate less
than twice that of the optimal scheme with half the paging memory.}

Due to the exponential scaling construction of $\mc{N}_\ell$, the number
of groups $L=\ceil{\log p_1/p_N}$ is usually small, i.e., the factor $L$
in the approximation gap is usually modest (see also the discussion in
Example~\ref{eg:zipf} below and in Section~\ref{sec:empirical}). We
illustrate this with several examples.

\begin{example}[\emph{Uniform File Popularity}]
    \label{eg:uniform}
    For the special case of uniform file popularities $p_n=1/N$
    for all $n\in \mc{N}$, we have $L=1$. Hence
    Theorems~\ref{thm:upper} and \ref{thm:lower} imply that the optimal
    expected rate $R^\star(M,\mc{N},K, \{1/N\})$ satisfies
    \begin{equation*}
        \frac{1}{c} r(M,N,K) 
        \leq R^\star(M,\mc{N},K, \{1/N\})
        \leq r(M,N,K),
    \end{equation*}
    showing that the peak and expected rates are approximately the same
    in this case and that the proposed coded caching scheme is within a
    constant factor of optimal. From the results in~\cite{maddah-ali13}
    this also implies that the expected rate of the scheme proposed here
    can be up to a factor $\Theta(K)$ smaller than HPF. Thus, we see
    that, while HPF minimizes the expected rate for a single cache
    ($K=1$), it can be significantly suboptimal for multiple caches
    ($K>1$). 
\end{example}

\begin{example}[\emph{Zipf File Popularity}]
    \label{eg:zipf}
    Consider next the important special case of a Zipf popularity
    distribution. This is a heavy-tail distribution with 
    \begin{align*}
        p_n & \defeq \zeta n^{-\alpha}, \\
        \zeta^{-1} & \defeq \sum_{n=1}^N n^{-\alpha},
    \end{align*}
    for all $n\in\mc{N}$ and for fixed parameter $\alpha > 0$ (see
    Fig.~\ref{fig:grouping} for an example). Typical values of the
    parameter $\alpha$ are between $1/2$ and $2$.

    In this case, there are several groups $\mc{N}_\ell$, and their sum
    popularities (i.e., the sum of the popularities of the files in
    $\mc{N}_\ell$) decay only slowly or not at all as a function of
    $\ell$. In fact, the cardinality $N_\ell$ of group $\mc{N}_\ell$ 
    is approximately 
    \begin{equation*}
        N_\ell \approx 2^{\ell/\alpha}
    \end{equation*}
    and their sum popularity is approximately 
    \begin{equation*}
        \sum_{n\in\mc{N}_\ell} p_n \approx \zeta 2^{\ell(1/\alpha-1)}.
    \end{equation*}
    Thus, we see that for $\alpha > 1$ the sum popularity decreases
    with $\ell$, whereas for $\alpha < 1$ it increases with $\ell$. 

    The proposed grouped coded caching scheme deals with this heavy tail
    by careful allocation of the cache memory among the different file
    groups. As a result, the scheme is able to exploit both the fact
    that \emph{individually} each file in a group $\mc{N}_\ell$ may have
    small probability, but at the same time \emph{collectively} each
    group of files may have high probability.

    Assume next that $\alpha > 1$. The expected number of users
    requesting files in group $\mc{N}_\ell$ is then approximately
    \begin{equation*}
        \E(\msf{K}_\ell) 
        = K \sum_{n\in\mc{N}_\ell} p_n 
        \approx  K \zeta 2^{\ell(1/\alpha-1)}.
    \end{equation*}
    Thus, for
    \begin{equation*}
        \ell \gg \frac{1}{1-1/\alpha}\log(\zeta K),
    \end{equation*}
    the expected number of users in group $\ell$ is less than one.  By
    serving the few users in groups larger than this by unicast
    delivery, we can effectively restrict the number $L$ of file groups
    to about $\frac{1}{1-1/\alpha}\log(\zeta K)$. As a
    result, Theorems~\ref{thm:upper} and \ref{thm:lower} then
    effectively (up to the unit additive term arising from users in the
    last group) provide a $\bigl(O(\log K), O(\log K)\bigr)$-bicriteria
    approximation guarantee for the proposed grouped coded caching scheme.
\end{example}

The proofs of Theorems~\ref{thm:upper} and \ref{thm:lower} are presented
in Appendices~\ref{sec:upper} and \ref{sec:lower}, respectively. The proof
of Theorem~\ref{thm:upper} analyzes the rate of the proposed scheme
using the results from~\cite{maddah-ali13} for each subgroup of files
and is straightforward. The proof of Theorem~\ref{thm:lower} links the
optimal expected rate to the optimal peak rate and is quite intricate,
involving a genie-based uniformization argument as well as a
symmetrization argument.

\section{Empirical Results}
\label{sec:empirical}

This section continues the comparison of the performance of the proposed
grouped coded caching scheme to HPF using an empirical file popularity
distribution. We choose the file popularities to be those of the $N =
10\,000$ most popular movies from the Netflix catalog. Following the
approach in~\cite{cha07}, we estimate the file popularities from the
dataset made available by Netflix for the Netflix Prize~\cite{netflix}.
The estimated file popularities are shown in Fig.~\ref{fig:netflix_p}. 

\begin{figure}[htbp]
    \centering%
    \includegraphics{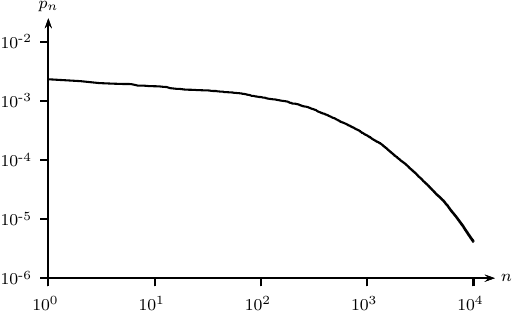}
    \caption{File popularities $\{p_n\}_{n=1}^N$ for the Netflix movie catalog.}
    \label{fig:netflix_p}
\end{figure}

As can be seen from the figure, the popularities exhibit a flat
``head'' consisting of the first 600 or so most popular files. This is
followed by a power-law ``tail'' with exponent of approximately $-2$.
This is in line with the behavior of other multimedia
content~\cite{gummadi03,hefeeda08}. 

We start with the analysis of HPF. For HPF, the expected rate is equal
to the expected number of users with a request outside the first $M$
most popular files, i.e., 
\begin{equation*} 
    K \sum_{n = M+1}^N p_n.
\end{equation*}
This rate is depicted in Fig.~\ref{fig:tradeoff} for $K=300$ users and
various values of cache size $M$. Increasing the cache size from $M$ to
$M+1$ decreases the rate of HPF over the shared link by $Kp_{M+1}$.
From Fig.~\ref{fig:netflix_p}, we expect the rate to initially decay
rather quickly with $M$ until the end of the ``head'' of the file
popularity curve. Once $M$ is big enough for the entire ``head'' to be
cached, we expect further decreases in $M$ to lead to diminishing
returns.  This behavior is indeed clearly visible in
Fig.~\ref{fig:tradeoff}. We conclude that a reasonable choice of $M$ for
HPF is thus the size of the ``head'' of the popularity distribution,
which in this case corresponds to about $M=600$.

\begin{figure}[htbp]
    \centering%
    \includegraphics{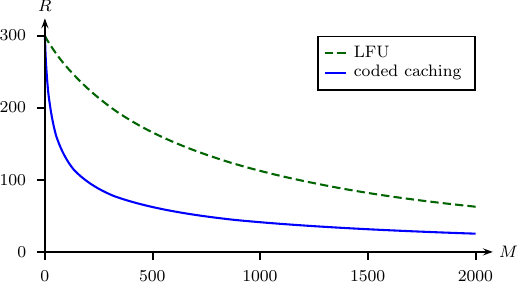}
    \caption{Memory-rate tradeoff under Netflix file popularities for
    the baseline HPF scheme (dashed green line) and the proposed grouped 
    coded caching scheme (solid blue line). The number of users is $K=300$,
    and the number of files is $N=10\,000$.}
    \label{fig:tradeoff}
\end{figure}

We continue with the evaluation of the proposed grouped coded caching
scheme.  From Theorem~\ref{thm:upper}, the rate of the proposed scheme
is
\begin{equation*}
    \min_{\{M_\ell\}: \sum_\ell M_\ell = M}
    \sum_{\ell=1}^L \E\bigl( r(M_\ell, N_\ell, \msf{K}_\ell) \bigr)
\end{equation*}
with $r(M, N, K)$ as defined in \eqref{eq:PeakRate}.  Note that $r(M,
N, K)$ is a concave function of $K$. We can thus apply Jensen's
inequality to upper bound the rate of the grouped coded caching scheme by 
\begin{equation}
\label{eq:RateJensen}
    \min_{\{M_\ell\}: \sum_\ell M_\ell = M}
    \sum_{\ell=1}^L r\bigl(M_\ell, N_\ell, \E(\msf{K}_\ell)\bigr),
\end{equation}
where 
\begin{equation}
\label{eq:RateJensen2}
    \E(\msf{K}_\ell) = K \sum_{n\in\mc{N}_\ell} p_n.
\end{equation}
We will be working with this upper bound in the following. This upper
bound on the rate of the proposed grouped coded caching scheme is depicted in
Fig.~\ref{fig:tradeoff}. 

Comparing the curves in Fig.~\ref{fig:tradeoff}, it is clear that the
proposed grouped coded caching scheme significantly improves upon the baseline
HPF scheme. In particular, for a cache size of $M=600F$ bits (where $F$
is the file size), HPF results in an average of $152F$ bits being sent
over the shared link. In contrast, for the same value of $M$, the
proposed scheme results in an average of $56F$ bits being sent over the
shared link---an improvement by more than a factor $2.7$.  Similarly,
assume we want to operate at the same expected load of $152F$ bits of the
shared link as achieved by HPF with $M=600F$ bits. The proposed coded
caching scheme can achieve the same load with only $M=63F$ bits in cache
memory---an improvement by a factor $9.5$.

From Theorems~\ref{thm:upper} and~\ref{thm:lower}, we also know that the
proposed coded caching scheme achieves the optimal memory-rate tradeoff
to within a factor $cL$ in the rate direction and to within a factor $L$
in the memory direction. In this example, the value of $L$ is $10$.

\section{Conclusions and Follow-Up Results}
\label{sec:conclusion}

In this paper, we have demonstrated and analyzed the benefits of coding
for caching systems with nonuniform file popularities. While (uncoded)
HPF is optimal for such systems with a single cache, we show here that
coding is required for the optimal operation of caching systems with
multiple caches.

Since a preprint of this work was first posted on arXiv in August of
2013, several follow-up papers~\cite{hachem14,ji15,hachem15,zhang15}
have refined the results presented here. In~\cite{ji15}, it is suggested
to use the decentralized coded caching scheme proposed
in~\cite{maddah-ali13} (see Algorithm~1 in this paper) for the $N_1$
most popular files, for some $N_1$. Any requests not from the $N_1$
most popular files are delivered directly from the server. This
corresponds to using the grouped coded caching scheme presented in this
paper with only $L=2$ groups and with $M_1=M$ and $M_2=0$. It is proved
in \cite{ji15} that by optimizing the value of $N_1$, this approach is
optimal to within a constant factor for Zipf popularity distributions in
the limit as $K, N \to \infty$. 

Subsequently, \cite{zhang15} suggested to choose $N_1$ as the largest
$n$ such that $KMp_{n} \geq 1$. By generalizing and tightening the
uniformization and symmetrization converse arguments developed in this
paper (in particular using a clever new argument to capture the load of
requests in the second, uncached, file group) it is shown in
\cite{zhang15} that this scheme is optimal to within a universal
constant multiplicative-plus-additive gap for \emph{all} popularity
distributions and all finite values of $K$ and $N$ (assuming $M \geq
2$). These two results \cite{ji15,zhang15} thus show that, surprisingly,
$L=2$ groups are sufficient to adapt to the nonuniform nature of the
popularity distribution.

In another line of follow-up work, \cite{hachem14,hachem15} takes a
different approach to deal with nonuniform popularity distributions. In
these works, it is assumed that files are split into several popularity
groups with a fixed number of users requesting files from each group.
For systems with a single user per cache (as considered here and in the
papers mentioned in the previous two paragraphs), \cite{hachem15} also
shows that it is approximately optimal to use only $L=2$ groups with all
memory allocated to the first group combined with decentralized coded
caching for the delivery. However, the situation changes for systems
with many users per cache, where \cite{hachem15} shows that grouped
coded caching with $L > 2$ groups is optimal to within a constant factor
in rate.

\appendices

\section{Optimality of the Highest-Popularity First Caching Rule for $K=1$}
\label{sec:HPF}

It is well known that HPF minimized the expected miss rate for systems
with a single cache ($K=1$) when coding is not allowed (see for example
\cite[Problem 6.11]{coffman73}).  In this appendix, we prove that for
single-cache systems HPF minimizes the expected rate among the larger
class of schemes that do allow coding. In particular, this shows that
coding is not required for $K=1$.

Recall that in the HPF placement phase each cache stores the $\floor{M}$
most popular files and the first $(M-\floor{M})F$ bits of file
$\floor{M}+1$. In the delivery phase, the server sends all uncached
requested files. If we assume that $p_1 \geq p_2 \geq \dots \geq p_N$ as
before, then HPF achieves expected rate $\sum_{n=\ceil{M}+1}^{N}
p_{n} + (\ceil{M} -M)  p_{\ceil{M}}$ for $K=1$. In particular, if $M$ is
an integer, then the expected rate of HPF is $\sum_{n=M+1}^{N}  p_{n} $. 

\begin{lemma}
    Highest-popularity first minimizes the expected rate for the
    $K=1$-user caching problem with $N$ files of arbitrary popularity. 
\end{lemma}

\begin{IEEEproof}
    We will show that any caching scheme has an expected rate at least
    as large as that of HPF. Assume the cache memories have been filled
    according to a placement function. In the delivery phase, the server
    sends the message $X_n$ of size $R_nF$ bits in response to the user
    requesting file $n$.  The user recovers file $n$ from $X_n$ and the
    content of its cache.  The caching scheme has then expected link
    load lower bounded as 
    \begin{align*}
        \sum_{n=1}^N p_n R_n F
        & \geq  \sum_{n=1}^N p_n H(X_n) \\
        &\overset{(a)}{=} \sum_{n=1}^{N} (p_{n}-p_{n+1}) \sum_{i=1}^{n} H(X_i)\\
        &\overset{(b)}{\geq} \sum_{n=1}^{N} (p_{n}-p_{n+1} ) (n-M)^+F\\
        & = \sum_{n=\ceil{M}}^{N} (p_{n}-p_{n+1} ) (n-M)F\\
        & = F\sum_{n=\ceil{M}}^{N} n p_{n} - F\sum_{n=\ceil{M}}^{N} n p_{n+1}  - FM\sum_{n=\ceil{M}}^{N} (p_{n}-p_{n+1} ) \\
        & = F\sum_{n=\ceil{M}}^{N} n p_{n} - F\sum_{n=\ceil{M}+1}^{N+1} (n-1) p_{n}  -FM p_{\ceil{M}}\\
        & = F\sum_{n=\ceil{M}+1}^{N}  p_{n}  + F(\ceil{M} -M)  p_{\ceil{M}},
    \end{align*}
    where in $(a)$ we set $p_{N+1}=0$.  The most important step in the
    above chain of inequalities is $(b)$, which is based on a cut-set
    argument as follows.  We note that the user can decode files $1, 2,
    \ldots, n$ from $X_1, X_2, \ldots, X_n$ and the content of its
    cache, and therefore $MF+\sum_{i=1}^{n}  H(X_i) \geq nF$. 
    To summarize, we have
    \begin{equation*}
        \sum_{n=1}^N p_n R_n 
        \geq \sum_{n=\ceil{M}+1}^{N}  p_{n}  + (\ceil{M} -M)  p_{\ceil{M}}
    \end{equation*}
    for any caching scheme.  Noting that the right-hand side is the
    expected rate achieved by HPF concludes the proof.
\end{IEEEproof}

\section{Proof of Theorem~\ref{thm:upper}}
\label{sec:upper}

We analyze the performance of the grouped coded caching scheme described
in Section~\ref{sec:theoretical}. The rate $r(M, N, K)$ as defined in
\eqref{eq:PeakRate} is the peak rate achieved by Algorithm~\ref{alg:1}
for $N$ files and $K$ users each with a cache memory of normalized size $M$. Thus,
the rate $r(M, N, K)$ of this scheme is the same for every possible user
requests. Put differently, the expected value (over all requests) of the
rate of Algorithm~\ref{alg:1} is the same as the rate for any specific
request.

Consider now a specific random request $(\msf{d}_1, \msf{d}_2, \dots,
\msf{d}_K)$.  As explained in Section~\ref{sec:theoretical}, this
request results in the users being partitioned into subsets $\mc{K}_1,
\dots, \mc{K}_L$ with cardinalities  $\msf{K}_1, \dots, \msf{K}_L$.
Since the delivery algorithm treats each of the groups $\mc{K}_\ell$
independently, the rate for request $(\msf{d}_1, \msf{d}_2, \dots,
\msf{d}_K)$ is
\begin{equation*}
    \sum_{\ell=1}^L r(M_\ell, N_\ell, \msf{K}_\ell).
\end{equation*}
We point out that the only randomness in this expression is due to the
random size $\msf{K}_\ell$ of the random group $\mc{K}_\ell$, which in
turn derives from the random user requests. Taking the expectation over
all $\msf{K}_\ell$ then yields the following upper bound on the optimal
expected rate $R^\star(M, \mc{N}, K, \{p_n\})$:
\begin{equation*}
    R^\star(M, \mc{N}, K, \{p_n\})
    \leq \sum_{\ell=1}^L \E\bigl( r(M_\ell, N_\ell, \msf{K}_\ell) \bigr).
\end{equation*}

We can minimize this upper bound by optimizing over the choice of memory
allocation. This yields
\begin{equation*}
    R^\star(M, \mc{N}, K, \{p_n\})
    \leq \min_{\{M_\ell\}: \sum_\ell M_\ell = M}
    \sum_{\ell=1}^L \E\bigl( r(M_\ell, N_\ell, \msf{K}_\ell) \bigr).
\end{equation*}
One particular choice of $M_\ell$ is $M/L$ for each $\ell$, which yields
\begin{equation*}
    R^\star(M, \mc{N}, K, \{p_n\})
    \leq 
    \sum_{\ell=1}^L \E\bigl( r(M/L, N_\ell, \msf{K}_\ell) \bigr).
\end{equation*}
Together, these two equations prove Theorem~\ref{thm:upper}.\hfill\IEEEQED

\section{Proof of Theorem~\ref{thm:lower}}
\label{sec:lower}

We will prove the equivalent statement
\begin{equation*}
    \sum_{\ell=1}^L \E\bigl( r(M, N_\ell, \msf{K}_\ell) \bigr)
    \leq c L R^\star(M,\mc{N},K, \{p_n\}).
\end{equation*}
The proof of this inequality is based on the following three claims. 

\begin{claim}
    \label{thm:claim1}
    For the caching problem with $N$ files $\mc{N}$ with \emph{uniform}
    popularity and with $K$ users each with normalized cache size $M$, we
    have
    \begin{equation*}
        r(M, N, K) \leq 72 R^\star(M, \mc{N}, K, \{1/N\}),
    \end{equation*}
    with $r(M, N, K)$ as defined in~\eqref{eq:PeakRate}.
\end{claim}

This claim upper bounds the peak rate $r(M, N, K)$ of
Algorithm~\ref{alg:1} by $72$ times the optimal expected rate
$R^\star(M, \mc{N}, K, \{1/N\})$ for the caching problem with
\emph{uniform} file popularity. Recall that Algorithm~\ref{alg:1} was
shown in~\cite{maddah-ali13} to be approximately optimal with respect to
the peak-rate criterion. The claim thus states that for uniform file
popularity optimal peak rate and optimal expected rate are approximately
the same. The proof of Claim~\ref{thm:claim1}, reported in
Appendix~\ref{sec:lower_claim1}, is based on a \emph{symmetrization
argument}, which aggregates the rates for several different demand
tuples and then applies a cut-set bound argument to this aggregated
rate. 

Applying Claim~\ref{thm:claim1} to the file set $\mc{N}_\ell$ of size
$N_\ell$ and with $K_\ell$ users, we thus have
\begin{equation}
    \label{eq:claim1}
    r(M, N_\ell, K_\ell) 
    \leq 72 R^\star(M, \mc{N}_\ell, K_\ell, \{1/N_\ell\}),
\end{equation}

\begin{claim}
    \label{thm:claim2}
    For the caching problem with $N$ files $\mc{N}$ of popularity
    $p_1\geq p_2\geq \dots \geq p_N$  satisfying $p_1 \leq 2p_N$ and
    with $K$ users each with normalized cache size $M$, we have
    \begin{equation*}
        R^\star(M, \mc{N}, K, \{1/N \})
        \leq 12 R^\star(M, \mc{N}, K, \{p_n\}).
    \end{equation*}
\end{claim}

This claim upper bounds the optimal expected rate $R^\star(M, \mc{N}, K,
\{1/N \})$ for a system with uniform file popularities by the optimal
expected rate $R^\star(M, \mc{N}, K, \{p_n\})$ for a system with almost
uniform file popularities (i.e., file popularities differing by at most
a factor two). Intuitively, the claim thus states that a small change in
the file popularity results only in a small change in the expected rate
of the optimal caching scheme. The proof of Claim~\ref{thm:claim2} is
reported in Appendix~\ref{sec:lower_claim2} and introduces a genie-based
\emph{uniformization argument} to transform almost uniform to uniform
file popularities.

Recall that the partition $\mc{N}_1, \mc{N}_2, \dots, \mc{N}_L$
guarantees that file popularities within the same file group differ by at
most a factor two. We can hence apply Claim~\ref{thm:claim2} to the file
set $\mc{N}_\ell$ with $K_\ell$ users to obtain
\begin{equation}
    \label{eq:claim2}
    R^\star(M, \mc{N}_\ell, K_\ell, \{1/N_\ell\})
    \leq 12 R^\star(M, \mc{N}_\ell, K_\ell, \{\xi_\ell p_n\} ),
\end{equation}
where the right-hand side is evaluated with respect to the file
popularities $\{\xi_\ell p_n\}_{n\in\mc{N}_\ell}$ for
normalization constant 
\begin{equation}
    \label{eq:xi}
    \xi^{-1}_\ell \defeq \sum_{n\in\mc{N}_\ell} p_n.
\end{equation}

\begin{claim}
    \label{thm:claim3}
    For every $\ell\in\{1,2,\dots,L\}$,
    \begin{equation*}
        \E\bigl( R^\star(M, \mc{N}_\ell, \msf{K}_\ell, \{\xi_\ell p_n\}) \bigr)
        \leq R^\star(M, \mc{N}, K, \{p_n\}),
    \end{equation*}
    where the left-hand side is evaluated with respect to the
    file popularities $\{\xi_\ell p_n\}_{n\in\mc{N}_\ell}$ for
    normalization constant $\xi_\ell$ as in~\eqref{eq:xi}.
\end{claim}

Recall that $R^\star(M, \mc{N}_\ell, k, \{\xi_\ell p_n\} )$ is the
optimal expected rate for a system with $k$ users.  Clearly, this is a
function of $k$, say $f(k)$. Let now $\msf{K}_\ell$ be the random number
of users in $\mc{K}_\ell$, and construct the random variable
$f(\msf{K}_\ell)$. Then the left-hand side of Claim~\ref{thm:claim3} is
the expectation of this random variable. Claim~\ref{thm:claim3} thus
states that if the server is only asked to handle the demands of users
in $\mc{K}_\ell$, ignoring the demands of the remaining users, the rate
of the optimal system decreases. The proof of Claim~\ref{thm:claim3} is
reported in Appendix~\ref{sec:lower_claim3}.

Using these three claims, Theorem~\ref{thm:lower} is now straightforward
to prove. Indeed, from Claims~\ref{thm:claim1} and~\ref{thm:claim2} (and
in particular \eqref{eq:claim1} and \eqref{eq:claim2}), we have for any
$K_1, \dots, K_L$, 
\begin{equation*}
    \sum_{\ell=1}^L r(M, N_\ell, K_\ell) 
    \leq 72\cdot 12 \sum_{\ell=1}^{L} R^\star(M, \mc{N}_\ell, K_\ell, \{\xi_\ell p_n\} ) .
\end{equation*}
Evaluating this expression at $K_\ell = \msf{K}_\ell$ and taking the
expectation yields
\begin{equation*}
    \sum_{\ell=1}^L \E\bigl( r(M, N_\ell, \msf{K}_\ell) \bigr)
    \leq 72\cdot 12 \sum_{\ell=1}^{L} \E\bigl( R^\star(M, \mc{N}_\ell,
    \msf{K}_\ell, \{\xi_\ell p_n\} ) \bigr).
\end{equation*}
Combining this with Claim 3 yields
\begin{equation*}
    \sum_{\ell=1}^L \E\bigl( r(M, N_\ell, \msf{K}_\ell) \bigr)
    \leq 72\cdot 12 L R^\star(M, \mc{N}, K, \{p_n\} ),
\end{equation*}
which proves the desired result with $c \defeq 72\cdot 12$. \hfill\IEEEQED

\subsection{Proof of Claim~\ref{thm:claim1} (Symmetrization and Cut-Set Arguments)}
\label{sec:lower_claim1}

We need to show that
\begin{equation}
    \label{eq:claim1_0}
    R^\star(M, \mc{N}, K, \{1/N\}) 
    \geq \frac{1}{72} r(M, N, K).
\end{equation}
The left-hand side is the expected rate of the optimal scheme for
uniform file popularity over $\mc{N}$, i.e., with $p_n=1/N$ for all
$n\in\mc{N}$. The right-hand side is (up to the constant) equal to the
peak rate \eqref{eq:PeakRate} of Algorithm~\ref{alg:1}. 

Let us first introduce some additional notation. Consider the random
demand vector $\underline{\msf{d}}\in\mc{N}^K$ and denote by
$w(\underline{\msf{d}})$ the number of its distinct entries.  For
$s\in\{1,2,\dots,\min\{N,K\}\}$, denote by $\bar{R}(M, \mc{N}, K,s)$ the
expected rate of the optimal scheme for uniform file popularity over
$\mc{N}$ when conditioned on the event that $w(\underline{\msf{d}})=s$.
We point out that in $\bar{R}(M, \mc{N}, K,s)$ both the placement phase
and the delivery phase of the system are optimized for this conditioning
on $w(\underline{\msf{d}})=s$.

\begin{example}
    \label{eg:repetition}
    For $s=K$, $\bar{R}(M, \mc{N}, K,K)$ corresponds to the expected
    rate of the optimal scheme with $K$ requests chosen uniformly at
    random from $\mc{N}$ \emph{without replacement}.
\end{example}

The proof of~\eqref{eq:claim1_0} consists of three steps, summarized by
the following three lemmas.

\begin{lemma}
    \label{thm:step1}
    For any $s\in\{1, 2, \dots, K\}$, we have
    \begin{equation*}
        R^{\star}(M, \mc{N}, K, \{1/N\})
        \geq \Pp\bigl(w(\underline{\msf{d}}) \geq s\bigr)\bar{R}(M,
        \mc{N}, s,s),
    \end{equation*}
    where $\underline{\msf{d}}$ is uniformly distributed over
    $\mc{N}^K$.
\end{lemma}

Lemma~\ref{thm:step1} lower bounds the expected rate of the optimal
scheme for uniform file popularities as a function of the expected rate
of the optimal scheme for file requests chosen uniformly at random without
replacement (see Example~\ref{eg:repetition}). 

\begin{lemma}
    \label{thm:step3}
    Assume $\underline{\msf{d}}$ is uniformly distributed over
    $\mc{N}^K$. Then, for $s \leq \ceil{\min\{N, K\}/4}$, we have
    \begin{equation*}
        \Pp\bigl(w(\underline{\msf{d}}) \geq s\bigr)
        \geq 2/3.
    \end{equation*}
\end{lemma}

Lemma~\ref{thm:step3} shows that with large probability the number of
distinct requests in $\underline{\msf{d}}$ is not too small.

\begin{lemma}
    \label{thm:step2}
    We have
    \begin{equation*}
        \max_{s\in\{1,\dots, \ceil{\min\{N,K\}/4}\}}\bar{R}(M, \mc{N}, s,s) 
        \geq \frac{1}{48} r(M, N, K).
    \end{equation*}
\end{lemma}

Lemma~\ref{thm:step2} is the key step in the proof
of~\eqref{eq:claim1_0}. It lower bounds the expected rate of the optimal
scheme for file requests chosen uniformly at random without replacement
as a function of the peak rate $r(M, N, K)$ of Algorithm~\ref{alg:1}.

Combining Lemmas~\ref{thm:step1}, \ref{thm:step3}, and \ref{thm:step2},
we obtain 
\begin{align*}
    R^{\star}(M, \mc{N}, K, \{1/N\})
    & \geq \max_{s \in \{1, \dots, \ceil{\min\{N,K\}/4}\}}
    \Pp\bigl(w(\underline{\msf{d}}) \geq s\bigr) \bar{R}(M, \mc{N}, s,s) \\
    & \geq \frac{2}{3}\cdot\frac{1}{48}r(M, N, K) \\
    & = \frac{1}{72}r(M, N, K),
\end{align*}
completing the proof of Claim~\ref{thm:claim1}. \hfill\IEEEQED

\begin{remark}
    Essentially the same argument as in Lemma~\ref{thm:step2}
    can be used to show that
    \begin{equation*}
        \bar{R}(M, \mc{N}, K,K) 
        \geq \frac{1}{12} r(M, N, K).
    \end{equation*}
    Thus, the peak rate $r(M, N, K)$ of Algorithm~\ref{alg:1} is within a
    factor of at most $12$ from the expected rate of the optimal scheme
    for requests chosen without replacement. 
\end{remark}

We next prove Lemmas~\ref{thm:step1}--\ref{thm:step2}.

\begin{IEEEproof}[Proof of Lemma~\ref{thm:step1}]
    Assume $\underline{\msf{d}}$ is uniformly distributed over
    $\mc{N}^K$, and denote by 
    \begin{equation*}
        q_j \defeq \Pp\bigl( w(\underline{\msf{d}}) = j \bigr)
    \end{equation*}
    the probability that the demand vector has exactly $j$ distinct
    entries. We can then lower bound the left-hand side in the statement of
    Lemma~\ref{thm:step1} as 
    \begin{align*}
      R^{\star}(M, \mc{N}, K, \{1/N\})
        & \geq \sum_{j=1}^K q_j \bar{R}(M, \mc{N}, K,j).
    \end{align*}

    Now, reducing the number of users can only decrease the optimal
    expected rate, so that
    \begin{equation*}
        \bar{R}(M, \mc{N}, K,j) \geq \bar{R}(M, \mc{N}, j,j).
    \end{equation*}
    Hence, we can further lower bound $R^\star(M, \mc{N}, K, \{1/N\})$ as
    \begin{align*}
        R^{\star}(M, \mc{N}, K, \{1/N\})
        & \geq \sum_{j=1}^K q_j \bar{R}(M, \mc{N}, j,j) \\
        & \geq \sum_{j=s}^K q_j \bar{R}(M, \mc{N}, j,j)
    \end{align*}
    for any $s \in \{1, 2, \ldots, K\}$. Similarly, 
    \begin{equation*}
        \bar{R}(M, \mc{N}, j,j) \geq \bar{R}(M, \mc{N}, s,s)
    \end{equation*}
    for $j \geq s$, so that
    \begin{align*}
        R^{\star}(M, \mc{N}, K, \{1/N\}) & \geq \sum_{j=s}^K q_j \bar{R}(M,
        \mc{N}, s,s) \\ & = \Pp\bigl( w(\underline{\msf{d}}) \geq s \bigr)
        \bar{R}(M, \mc{N}, s,s)
    \end{align*}
    for any $s \in \{1, 2, \ldots,K\}$.
\end{IEEEproof}

\begin{IEEEproof}[Proof of Lemma~\ref{thm:step3}]
    Consider a sequence of independent and identically distributed random
    variables $\msf{d}_1, \msf{d}_2, \dots$ each uniformly distributed over
    $\mc{N}$, and let
    \begin{equation*}
        \underline{\msf{d}} \defeq (\msf{d}_k)_{k=1}^K.
    \end{equation*}
    For $s \leq\ceil{\min\{N,K\}/4}$, we aim to lower bound
    \begin{align*}
        \Pp\bigl(w(\underline{\msf{d}}) \geq s \bigr) 
        & \geq \Pp\bigl(w(\underline{\msf{d}}) \geq \ceil{\min\{N,K\}/4}\bigr).
    \end{align*}
    This is a standard coupon collector problem, and our analysis 
    follows~\cite[Chapter 3.6]{motwani95}.
    
    Let 
    \begin{equation*}
        \msf{f}_k \defeq w\bigl((\msf{d}_1, \dots, \msf{d}_k)\bigr)
    \end{equation*}
    be the number of distinct elements in the first $k$ requests.
    Note that $\msf{f}_1, \msf{f}_2, \dots$ is an increasing sequence of
    random variables. We denote by the random variable $\msf{z}_i$ the
    number of elements in the random sequence $\msf{f}_1, \msf{f}_2, \ldots$
    that take value $i$. Observe that $\msf{z}_1, \msf{z}_2, \dots$ are
    independent random variables, and $\msf{z}_i$ is geometrically
    distributed with parameter $(N-i+1)/N$. 
    
    Set
    \begin{equation*}
        \msf{z} \defeq \sum_{i=1}^{\ceil{\min\{N,K\}/4}-1} \msf{z}_i.
    \end{equation*}
    Note that $\msf{z}=k$ means that $k+1$ is the first time such that
    $w\bigl((\msf{d}_1, \dots, \msf{d}_{k+1})\bigr) = \ceil{\min\{N,K\}/4}$. Hence,
    \begin{equation*}
        \Pp\bigl(w(\underline{d}) \geq \ceil{\min\{N,K\}/4}\bigr) 
        = \Pp(\msf{z} < K).
    \end{equation*}
    
    Now
    \begin{align*}
        \E(\msf{z}) 
        & = \sum_{i=1}^{\ceil{\min\{N,K\}/4}-1} \frac{N}{N-i+1} \\
        & \leq \bigl(\ceil{\min\{N,K\}/4}-1\bigr) \frac{N}{N-\ceil{\min\{N,K\}/4}+2} \\
        & \leq \frac{\min\{N,K\}}{4}\cdot\frac{N}{N-\min\{N,K\}/4} \\
        & \leq \frac{K}{4}\cdot\frac{N}{3N/4} \\
        & = K/3.
    \end{align*}
    Hence, by Markov's inequality,
    \begin{equation*}
        \Pp(\msf{z} \geq K) 
        \leq \frac{E(\msf{z})}{K}
        \leq 1/3.
    \end{equation*}
    This implies that
    \begin{equation*}
        \Pp\bigl(w( \underline{\msf{d}}) \geq \ceil{\min\{N,K\}/4}\bigr) 
        = 1-\Pp(\msf{z} \geq K) 
        \geq 2/3,
    \end{equation*}
    proving Lemma~\ref{thm:step3}. 
\end{IEEEproof}

\begin{IEEEproof}[Proof of Lemma~\ref{thm:step2}]
    We make use of a symmetrization argument combined with a cut-set
    argument around $s$ users. Fix a value of $s\in\{1, 2, \dots, \ceil{\min\{N,K\}/4}\}$.  

    We start with the symmetrization argument. Consider a
    scheme achieving the optimal expected rate $\bar{R}(M, \mc{N}, s, s)$
    for a system with $s$ users requesting files uniformly at random
    from $\mc{N}$ without replacement. For a particular request
    $\underline{d}\in\mc{N}^s$ with $s$ distinct entries
    $w(\underline{d}) = s$, denote by $\bar{R}_{ \underline{d}}(M, \mc{N}, s,
    s)$ the rate of this scheme when the request vector is
    $\underline{d}$. Since there are $N!/(N-s)!$ different such
    request vectors, each with equal probability, we have
    \begin{equation}
        \label{eq:claim1_1}
        \bar{R}(M, \mc{N}, s,s)
        = \frac{(N-s)!}{N!}\sum_{\underline{d} \in\mc{N}^s: w(\underline{d}) = s} 
        \bar{R}_{ \underline{d}}(M, \mc{N}, s,
    s).
    \end{equation}

    Now, let
    \begin{equation}
        \label{eq:idef}
        I \defeq \floor{N/s},
    \end{equation}
    and consider $I$-tuples $(\mc{S}_1, \mc{S}_2, \dots, \mc{S}_I)$ of
    subsets of $\mc{N}$ with the property that each subset
    $\mc{S}_i\subset\mc{N}$ has cardinality $s$ and that distinct
    subsets are disjoint. By the definition of $I$ such subsets exist.
    Denote by $\mc{P}$ the collection of all possible such ordered
    $I$-tuples $(\mc{S}_1, \mc{S}_2, \dots, \mc{S}_{I})$. Note that, by
    symmetry, every possible subset $\mc{S}$ of cardinality $s$ is
    contained the same number of times in $I$-tuples in $\mc{P}$---this
    is the key property resulting from the symmetrization construction.
    Let $B$ be that number. We can then rewrite 
    \begin{align}
        \label{eq:claim1_2}
        \sum_{\underline{d} \in\mc{N}^s: w(\underline{d}) = s}
        \bar{R}_{ \underline{d}}(M, \mc{N}, s,s)     
           & = \sum_{\mc{S}\subset\mc{N}: \card{\mc{S}}=s}\;
        \sum_{\underline{d} \in\mc{S}^s: w(\underline{d}) = s} 
        \bar{R}_{ \underline{d}}(M, \mc{N}, s, s) \notag\\
        & = \frac{1}{B} \sum_{(\mc{S}_1, \dots, \mc{S}_I)\in\mc{P}}\; \sum_{i=1}^I \; 
        \sum_{\underline{d}\in \mc{S}_i^s: w(\underline{d}) = s} 
        \bar{R}_{ \underline{d}}(M, \mc{N}, s, s)  \notag\\
        & = \frac{1}{B} \sum_{(\mc{S}_1, \dots, \mc{S}_I)\in\mc{P}} \; 
        \sum_{\substack{(\underline{d}_1,\dots,\underline{d}_I)\in
        (\mc{S}_1^s, \dots, \mc{S}_I^s): \\ w(\underline{d}_i) = s \;\forall i}} \;
        \sum_{i=1}^I 
 \bar{R}_{ \underline{d}_i}(M, \mc{N}, s, s).
    \end{align}

    Fix $(\mc{S}_1, \mc{S}_2, \dots, \mc{S}_I)$ in $\mc{P}$ and consider
    corresponding demand vectors $(\underline{d}_1, \underline{d}_2, \dots,
    \underline{d}_I)$, where $\underline{d}_i\in\mc{S}_i^s$ with
    $w(\underline{d}_i) = s$.  We next use a cut-set argument to lower bound
    the sum
    \begin{equation*}
        \sum_{i=1}^{I}  \bar{R}_{ \underline{d}_i}(M, \mc{N}, s, s).
    \end{equation*}

    Recall that $\bar{R}_{ \underline{d}_i}(M, \mc{N}, s, s)$ is the
    rate of a system with $s$ users. Consider those $s$ users. From the
    content of their caches (of total size $sMF$ bits) and the $I$
    transmissions (of total size  $\sum_{i=1}^{I}
    \bar{R}_{\underline{d}_i}(M, \mc{N}, s, s)F$ bits), these users
    together are able to recover the $sI$ distinct files $\cup_{i=1}^I
    \mc{S}_i$  (of total size $sIF$ bits). Hence, by the cut-set bound,
    we must have
    \begin{equation*}
        sMF + \sum_{i=1}^I \bar{R}_{ \underline{d}_i}(M, \mc{N}, s, s)F 
        \geq sIF.
    \end{equation*}
    Simplifying this expression, we obtain that
    \begin{equation*}
        \sum_{i=1}^I \bar{R}_{ \underline{d}_i}(M, \mc{N}, s, s)  \geq s(I-M).
    \end{equation*}
    Since the left-hand side of this inequality is always nonnegative, this can
    be sharpened to 
    \begin{equation*}
        \sum_{i=1}^I \bar{R}_{ \underline{d}_i}(M, \mc{N}, s, s) \geq s(I-M)^+,
    \end{equation*}
    where $(x)^+$ denotes $\max\{x, 0\}$.  Combining this
    with~\eqref{eq:claim1_2}, we can lower bound the right-hand side of
    ~\eqref{eq:claim1_1} as
    \begin{equation}
        \label{eq:claim1_1_second}
        \frac{(N-s)!}{N!}
        \sum_{\underline{d} \in\mc{N}^s: w(\underline{d}) = s} 
        \bar{R}_{ \underline{d}_i}(M, \mc{N}, s, s)
        \geq \frac{1}{I}\cdot s(I-M)^+,
    \end{equation}
    where the normalization $1/I$ arises because we have lower bounded the
    sum of $I$ terms at a time. 

    Substituting \eqref{eq:idef} and \eqref{eq:claim1_1_second} into
    \eqref{eq:claim1_1} yields
    \begin{equation*}
       \bar{R}(M, \mc{N}, s, s)
        \geq s\biggl(1-\frac{M}{\floor{N/s}}\biggr)^+,
    \end{equation*}
    so that
    \begin{equation}
        \label{eq:claim1_3}
        \max_{s \in \{1, \dots, \ceil{\min\{N,K\}/4}\}}
         \bar{R}(M, \mc{N}, s, s)
        \geq \max_{s \in \{1, \dots, \ceil{\min\{N,K\}/4}\}}
        s\biggl(1-\frac{M}{\floor{N/s}}\biggr)^+.
    \end{equation}

    We continue by analyzing the right-hand side. We claim that
    \begin{equation*}
        \max_{s \in \{1, \dots, \ceil{\min\{N,K\}/4}\}}
        s\biggl(1-\frac{M}{\floor{N/s}}\biggr)^+ 
        \geq \frac{1}{4} \max_{s \in \{1, \dots, \min\{N,K\}\}} 
        s\biggl(1-\frac{M}{\floor{N/s}}\biggr)^+.
    \end{equation*}
    Let $s^\star$ be the maximizer of the right-hand side.  If $s^\star \leq
    \ceil{\min\{N,K\}/4}$, then clearly the inequality holds. Assume then that
    $\ceil{\min\{N,K\}/4} < s^\star \leq \min\{N,K\}$. Then
    \begin{align*}
        s^\star\biggl(1-\frac{M}{\floor{N/s^\star}}\biggr)^+
        & \leq 4 \min\{N,K\}/4
        \biggl(1-\frac{M}{\floor{N/\ceil{\min\{N,K\}/4}}}\biggr)^+ \\
        & \leq 4 \ceil{\min\{N,K\}/4} 
        \biggl(1-\frac{M}{\floor{N/\ceil{\min\{N,K\}/4}}}\biggr)^+ \\
        & \leq 4 \max_{s\in\{1, \dots, \ceil{\min\{N,K\}/4}\}} 
        s\biggl(1-\frac{M}{\floor{N/s}}\biggr)^+,
    \end{align*}
    and the inequality holds as well. Together with~\eqref{eq:claim1_3},
    this shows that
    \begin{equation*}
        \max_{s \in \{1, \dots, \ceil{\min\{N,K\}/4}\}}
         \bar{R}(M, \mc{N}, s, s)
        \geq \frac{1}{4} \max_{s \in \{1, \dots, \min\{N,K\}\}}
        s\biggl(1-\frac{M}{\floor{N/s}}\biggr)^+.
    \end{equation*}

    Now, we know by~\cite[Theorem~2]{maddah-ali13} that
    \begin{equation*}
        \max_{s \in \{1, \dots, \min\{N,K\}\}} 
        s\biggl(1-\frac{M}{\floor{N/s}}\biggr)^+
        \geq \frac{1}{12}r(M, N, K).
    \end{equation*}
    Combining these last two inequalities, we obtain that
    \begin{equation*}
        \max_{s \in \{1, \dots, \ceil{\min\{N,K\}/4}\}}
         \bar{R}(M, \mc{N}, s, s) 
        \geq \frac{1}{48}r(M, N, K)
    \end{equation*}
    as needed to be shown.
\end{IEEEproof}

\subsection{Proof of Claim~\ref{thm:claim2} (Uniformization Argument)}
\label{sec:lower_claim2}

We need to show that, if $1/2 \leq p_N/p_n \leq 1$ for all
$n\in\mc{N}$, then
\begin{equation*}
    R^\star(M, \mc{N}, K, \{p_n\}) 
    \geq \frac{1}{12}R^\star(M, \mc{N}, K, \{1/N\}).
\end{equation*}
The left-hand side is the optimal expected rate for a
system with $K$ users requesting the files $\mc{N}$ with popularities
$\{p_n\}_{n\in\mc{N}}$.  The right-hand side is (up to the constant)
the optimal expected rate for the same system but with uniform file popularities.

The uniformization argument is as follows.  Assume that, at the
beginning of the delivery phase of the system, a genie arrives to aid
the transmission of the files.  Consider a user requesting
file $n$. The genie flips a biased coin yielding head with probability
$p_N/p_n \geq 1/2$. If the coin shows tail, the genie provides the user
the requested file for free. If the coin shows head, he does not help
the user. Thus, the probability that a user requests file $n$ and is
not helped by the genie is equal to 
\begin{equation*}
    p_n \cdot \frac{p_N}{p_n} = p_N.
\end{equation*}
Observe that this probability is the same for each file $n$. The genie
repeats this procedure independently for each user.

The users that have their file delivered by the genie can be ignored in
the subsequent delivery phase. The resulting system is thus one with a
random number $\tilde{\msf{K}}$ of users requesting one of the files in
$\mc{N}$ with uniform probability. In other words, we have transformed the
original problem with a fixed number $K$ of users with nonuniform file
popularities into a new problem with a random number of users with
uniform file popularities. 

Consider the scheme achieve the optimal expected rate $R^\star(M,
\mc{N}, \tilde{K}, \{p_n\})$, and let $R_{ \underline{d}}^\star(M,
\mc{N}, \tilde{K}, \{p_n\})$ the rate of that scheme for the request
vector $\underline{d}$. Then, we have
\begin{align}
    \label{eq:claim2_1}
    R^\star(M, \mc{N}, K, \{p_n\})  
    & \geq \sum_{\tilde{K}=1}^K \Pp(\tilde{\msf{K}} = \tilde{K})
    \sum_{\underline{d} \in\mc{N}^{\tilde{K}}} N^{-\tilde{K}}
    R_{ \underline{d}}^\star(M, \mc{N}, \tilde{K}, \{p_n\}) \nonumber\\ 
    & \geq \sum_{\tilde{K}=1}^K \Pp(\tilde{\msf{K}} = \tilde{K})
    R^{\star}(M, \mc{N}, \tilde{K}, \{1/N\}),
\end{align}
where the first inequality follows from the genie-aided argument and the second
inequality follows since $ R^{\star}(M, \mc{N}, \tilde{K}, \{1/N\})$ is
the optimal expected rate under uniform file popularities.

Consider the number of users $K-\tilde{\msf{K}}$ that are helped by the genie.
Recall that the probability $1-p_N/p_n$ that the genie helps is upper bounded
by $1/2$ by assumption on $\{p_n\}_{n\in\mc{N}}$. We therefore have 
\begin{equation*}
    \E(K-\tilde{\msf{K}}) \leq K/2. 
\end{equation*}
By Markov's inequality, we thus obtain
\begin{equation*}
    \Pp(K-\tilde{\msf{K}} \geq 3K/4) \leq 2/3.
\end{equation*}
From this,
\begin{equation*}
    \Pp(\tilde{\msf{K}} \geq \ceil{K/4}) \geq 1/3.
\end{equation*}
In words, with probability at least $1/3$ there are at least
$\ceil{K/4}$ users that are not helped by the genie.

Using this inequality, the right-hand side of \eqref{eq:claim2_1} can be
further lower bounded as
\begin{align}
    \label{eq:claim2_2}
    \sum_{\tilde{K}=1}^K \Pp(\tilde{\msf{K}} = \tilde{K})
     R^{\star}(M, \mc{N}, \tilde{K}, \{1/N\})
    & \geq \sum_{\tilde{K} \geq \ceil{K/4}} \Pp(\tilde{\msf{K}} = \tilde{K}) 
    R^{\star}(M, \mc{N}, \tilde{K}, \{1/N\}) \notag\\
    & \geq \Pp\bigl(\tilde{\msf{K}} \geq \ceil{K/4}\bigr) 
     R^{\star}(M, \mc{N}, \ceil{K/4}, \{1/N\}) \notag\\
    & \geq \frac{1}{3}  R^{\star}(M, \mc{N}, \ceil{K/4}, \{1/N\}).
\end{align}
Notice that the right-hand side is $1/3$ of the optimal expected rate 
for a system with $\ceil{K/4}$ users. 

We would like to relate this to the optimal expected rate 
for a system with $K$ users.  Take such a system and partition the $K$
users into four subsets each with at most $\ceil{K/4}$ users. We can
treat these four subsets of users as parallel systems, in which case the
delivery rate is the sum of the delivery rates for each of the four
parallel systems.  Since the optimal scheme can be no worse than this,
we have the inequality
\begin{equation*}
    R^{\star} (M, \mc{N}, K, \{ 1/N\}) 
    \leq 4 R^{\star} (M, \mc{N}, \ceil{K/4}, \{ 1/N\}) .
\end{equation*}
Using this, \eqref{eq:claim2_2} is further lower bounded as
\begin{equation}
    \label{eq:claim2_3}
    \frac{1}{3} 
    R^{\star} (M, \mc{N}, \ceil{K/4}, \{ 1/N\}) 
    \geq \frac{1}{12} R^{\star} (M, \mc{N}, K, \{ 1/N\}) .
\end{equation}

Combining \eqref{eq:claim2_1}--\eqref{eq:claim2_3}  yields that
\begin{equation*}
   R^{\star}(M, \mc{N}, K, \{p_n\})
    \geq \frac{1}{12} R^{\star}(M, \mc{N}, K, \{1/N\}),
\end{equation*}
proving the claim. \hfill\IEEEQED

\subsection{Proof of Claim~\ref{thm:claim3}}
\label{sec:lower_claim3}

We will show that
\begin{equation*}
    R^\star(M, \mc{N}, K, \{p_n\}) \geq
    \E\bigl( R^\star(M, \mc{N}_\ell, \msf{K}_\ell, \{\xi_\ell p_n\}  ) \bigr).
\end{equation*}
The left-hand side is the expected rate of the optimal scheme for the
original caching problem with $K$ users and files $\mc{N}$. Now assume
that, at the beginning of the delivery phase of the system, a genie
provides to each user requesting a file outside $\mc{N}_\ell$ the
requested file for free. Clearly, this can only reduce the rate over the
shared link. The right-hand side is a lower bound on the expected rate
of the optimal scheme for this genie-aided system. \hfill\IEEEQED

\section*{Acknowledgment}

The authors thank the reviewers for their careful reading of the
manuscript and their detailed comments.

\end{document}